\documentclass[10pt,conference]{IEEEtran}
\IEEEoverridecommandlockouts
\usepackage{cite}
\usepackage{indentfirst}
\usepackage{multirow}
\usepackage{booktabs}
\usepackage{makecell}
\usepackage{float}
\usepackage{amsmath,amssymb,amsfonts}
\usepackage{algorithmic}
\usepackage{graphicx}
\usepackage{graphicx,subfigure}
\graphicspath{{./figures/}}
\usepackage{textcomp}
\usepackage{xcolor}
\usepackage{subfigure}
\usepackage{booktabs}
\usepackage{bm}
\usepackage{amsthm}
\def\BibTeX{{\rm B\kern-.05em{\sc i\kern-.025em b}\kern-.08em
		T\kern-.1667em\lower.7ex\hbox{E}\kern-.125emX}}
	
\begin{document}
	\bibstyle{IEEEtran}
	\title{Mobile MIMO Channel Prediction with ODE-RNN: a Physics-Inspired Adaptive Approach}
	
	
	\author{
		\IEEEauthorblockN{Zhuoran~Xiao,
			Zhaoyang~Zhang$\IEEEauthorrefmark{2}$,
			Zirui~Chen,
			Zhaohui~Yang,
			and Richeng~Jin
}
		
		\IEEEauthorblockA{
			College of Information Science and Electronic Engineering, Zhejiang University, Hangzhou, China\\
			Zhejiang Provincial Key Laboratory of Info. Proc., Commun. \& Netw. (IPCAN), Hangzhou, China\\
			E-mails: \{zhuoranxiao, ning\_ming$\IEEEauthorrefmark{2}$,
			zirui\_chen, richengjin\}@zju.edu.cn, zhaohui.yang@ucl.ac.uk
		}
	}
\maketitle

\begin{abstract}
Obtaining accurate channel state information (CSI) is crucial and challenging for multiple-input multiple-output (MIMO) wireless communication systems. Conventional channel estimation method cannot guarantee the accuracy of mobile CSI while requires high signaling overhead. Through exploring the intrinsic correlation among a set of historical CSI instances randomly obtained in a certain communication environment, channel prediction can significantly increase CSI accuracy and save signaling overhead. In this paper, we propose a novel channel prediction method based on ordinary differential equation (ODE)-recurrent neural network (RNN) for accurate and flexible mobile MIMO channel prediction. Differing from existing works using sequential network structures for exploring the numerical correlation between observed data, our proposed method tries to represent the implicit physics process of path responses changing by specially designed continuous learning network with ODE structure. Due to the targeted design of learning network, our proposed method fits the mathematics feature of CSI data better and enjoy higher network interpretability. Experimental results show that the proposed learning approach outperforms existing methods, especially for long time interval of the CSI sequence and large channel measurement error.

\end{abstract}


\section{Introduction} \label{intro}

Obtaining accurate channel state information (CSI) is of crutial importance to wireless communication systems. With the help of CSI, the base station (BS) can adaptively perform resource scheduling, such as adjusting the modulation order, transmission power and precoding codeword to achieve performance gain \cite{9427230}. In general, CSI is obtained at the receiver by channel estimation algorithm, and then feedback to the BS. There exist at least two potential drawbacks for this approach. On one hand, if the channel parameters to be estimated are strictly limited, the accuracy of the estimated channel cannot be guaranteed. Meanwhile, if the number of parameters to be estimated increases, the signaling overhead and communication delay undoubtly increase \cite{8395053}. 
On the other hand, there exists nonnegligible transmission delay for the BS to receive feedback CSI data. 
Besides, the instantaneous channel has changed when the BS received the feedback CSI data, which will inevitably cause the loss of accuracy.

In practical scenarios, the BS often serves a fixed area. Also, the historical channel and the instantaneous channel to be predicted are both affected by the same environment scatterers. Therefore, the instantaneous channel has strong temporal and spatial correlation with the historic channel obtained in the past period of time. This correlation motivates us to make channel prediction by making full use of the channel instances obtained in the past. There are at least two benefits to adopt channel prediction. On one hand, combining channel estimation with the predicted channel, more accurate instantaneous CSI can be obtained. On the other hand, signaling overhead and processing delay can be greatly reduced since fewer pilots are needed.

Through statistically modeling a wireless channel as a set of radio propagation parameters, two conventional prediction approaches, namely parametric model \cite{6945858} and auto-regressive model \cite{2000Long} have been proposed. Due to the gap between the conventional prediction model and real wireless channel, the statistically modeling-based prediction is generally inaccurate and infeasible in practical systems. Some recent works applied machine learning algorithms to predict the current and following CSI based on a series of past CSI sequence. In \cite{2002Recurrent,2014Fading}, recurrent neural network (RNN) is proposed to build a predictor for narrow-band single-antenna channels. Besides, RNN is further replaced by a long short-term memory (LSTM) and gated recurrent unit (GRU) in \cite{2020Recurrent}. The authors in \cite{8904286} further adopt the sequence to sequence structure and use a generation model to predict the channel. 

Although the existing learning based methods for channel prediction can achieve better prediction accuracy than traditional methods, the prediction accuracy and other performance metrics are still insufficient for practical application. The accuracy of time-sequence prediction by a purely data-driven learning network is highly relied on the numerical correlation and smoothness of the sequence data. Thus, the prediction accuracy of those methods will decrease significantly when the interval of sequential data increases. Besides, these methods require strict equal interval of the obtained CSI sequence. However, some observed results may have large errors due to poor channel state in practical, which will greatly affect the system performance. The reason behind is that the network structures adopted in existing works mainly focus on the data correlation between discrete observed data while ignoring the fact that the practical channel changing is a continuous process with unique physical properties which is need to be fully fully considered in the learning network design.

To solve the problems mentioned above, we introduce the neural ordinary differential equation (Neural ODE) to implicitly represent the physical processes of channel changing rather than simply exploring the numerical correlations of discrete observed data. Furthermore, ODE-RNN is adopted to replace the original Neural ODE, overcoming the structure drawback that the whole CSI sequence data cannot be fully utilized. Moreover, the computational cost of this learning structure is quite low which means low calculation delay, making it quite suitable for channel prediction in practical.

The remainder of this paper is organized as follows. The system  model is described in section \ref{system}. The motivation for our network design is given in section \ref{net}. Section \ref{scene} introduces our experiment scene setup. Numerical results which evaluate our proposed approach compared with existing methods from different perspective are provided in section \ref{performance}. Section \ref{conclusion} draws our main conclusions.

\section{System Model} \label{system}
\subsection{Channel Model}
We consider that the BS is equipped with a multi antenna array with half wavelength spacing between two adjacent antennas and adopting orthogonal frequency-division multiplexing (OFDM) modulation. User equipment (UE) has a single omni-directional antenna.
The BS has ${N_t}$ antennas and there are ${N_c}$ subcarriers for OFDM signals. 
The channel frequency response (CFR) for each subcarrier can be formulated as
\begin{equation}\label{Sye_CFR_eq1}
{\bf{h}}[l] = \sum\limits_{p = 1}^K {{\alpha _p}{\bf{e}}({\theta _p})} {e^{ - j2\pi [{{{d_p} + {v_u}\cos ({\theta _v} - {\theta _p}){\tau _p}} \over {{\lambda _l}}}]}},
\end{equation}
where $l$ denotes the subcarrier index, $K$ is the total number of propagation paths, ${{\alpha _p}}$ is the path loss of $p$th path, ${{\theta_p}}$ is the angle of arrival, ${{d_p}}$ is the length of propagation path, ${{v_u}}$ denotes the velocity of user, ${{\theta _v}}$ is the direction angle of the velocity vector, ${{\tau _p}}$ is the propagation delay and ${{\lambda _l}}$ is the wavelength of corresponding subcarrier. In equation \eqref{Sye_CFR_eq1}, ${\bf e}(\theta)$ denotes the array response vector of the ULA given by
\begin{equation}
	{\bf{e}}(\theta ) = {[1,{e^{ - j2\pi {{d\cos (\theta )} \over \lambda }}},...,{e^{ - j2\pi {{({N_t} - 1)d\cos (\theta )} \over \lambda }}}]^T},
\end{equation}
where $d$ is the gap between two adjacent antennas and $\lambda $ is the wavelength. Thus, the overall CFR matrix of the channel between the BS and the user can be expressed as 
\begin{equation}
{\bf{H}} = [{\bf{h}}[1],{\bf{h}}[2],...,{\bf{h}}[{N_c}]].
\end{equation}
Note that the matrix ${\bf{H}}$ is referred as CSI matrix in the literature.

\subsection{Problem Formulation}
The goal of channel prediction is to forecast the CSI at the current and following time by taking full use of CSI sequence information obtained in the previous time period. In existing works \cite{2020Recurrent,8904286}, the CSI data needs to be uniformly sampled. 
Assume that the BS stores CSI estimated in the past $n$ time slots, which can be denoted by $\{ {\bf{H}}[1],{\bf{H}}[2],...,{\bf{H}}[n]\} $. The CSI in the next time slot should be predicted, denoted by $\mathop {\bf{H}}\limits^ \wedge  [n + 1]$. 
Thus, the CSI prediction problem can be presented as 
\begin{equation}
	\{ {\bf{H}}[1],{\bf{H}}[2],...,{\bf{H}}[n]\}  \to \mathop {\bf{H}}\limits^ \wedge  [n + 1].
\end{equation}

As one advantage of our proposed method, the CSI sequence is no longer required to be evenly sampled. Thus, the process can be written as
\begin{equation}
\{ {\bf{H}}[{t_1}],{\bf{H}}[{t_2}],...,{\bf{H}}[{t_n}]\}  \to \mathop {\bf{H}}\limits^ \wedge  [{t_{n + 1}}],
\end{equation}
where $t_x$ is a certain time point and the time sequence $t_1, t_2, \cdots, t_n, t_{n+1}$ can be chosen as any arbitrary sequence.

\section{ODE-RNN for Mobile Channel Prediction: Motivation and Learning Structure} \label{net}

\subsection{Overview of Neural ODE and ODE-RNN}
Due to its unique structural design, Neural ODE has been proved to have a strong ability to represent and predict the time series driven by physics processes. In a typical Neural ODE, the state of hidden layer is defined by the solution of the following equation,
\begin{equation}
{{dx(t)} \over {d t}} = f({\bf x}(t),{\bf{I}}(t),t,{\boldsymbol\theta} ),
\end{equation}
where ${\bf x}(t) \in {\mathbb{R}^D}$ denotes the hidden layer, $D$ is the dimension of variables in the hidden layer, ${\bf{I}}(t)$ is the input, $t$ denotes time, $f( \cdot )$ denotes a neural network with parameter $\boldsymbol\theta$. Therefore, after the network completes training, the forward calculation becomes an ordinary differential equation problem with known initial values. Thus, the forward calculation can be solved by numerical methods which is called ODE Solver. The process of  ODE Solver can be written as 
\begin{equation}
h({t_0}) = {h_0},
\end{equation}
\begin{equation}
{h_0},...,{h_N} = \text{ODE Solver}(f_{\boldsymbol\theta},{h_0},({t_0},...,{t_N})).
\end{equation}
Among all the numerical methods, Euler method is one of the most commonly used ODE Solvers. The procedure of Euler method can be presented by 
\begin{equation}
{h_{t + \Delta t}} = {h_t} + \Delta t \times f({h_t},{\theta _t}),
\end{equation}
where $\Delta t$ is the step length used to adjust the accuracy and computation cost. In addition to the Euler method, there are also some high-order solvers with adaptive step size. Choosing a proper ODE Solver is determined by the trade-off between computational cost and accuracy. 

In order to ensure high calculation accuracy, the step size of ODE Solver is actually set as a  small value, which indicates that directly using the gradient back propagation algorithm to calculate the loss function will introduce a large calculation cost. To solve this problem, the adjoint method is proposed in \cite{2018Optimization} to transform the gradient calculation into an ODE problem, which can be solved by the ODE Solver with low computational cost.
To handle the calculation error introduced in the adjoint method in \cite{2018Optimization}, 
the adaptive checkpoint adjoint method is further proposed in \cite{2020Adaptive} with adding a small amount of storage cost. 

For a continuous system defined by Neural ODE, after all the network parameters are fixed through training, the change of hidden state over time is only determined by the initial input of the network. However, the available data is actually the CSI sequence obtained at a previous time period. The forward calculation of Neural ODE can only take use of one data as the initial value, while other observations cannot participate in the process. Obviously, this highly limits the ability for fully digging the correlation between sequence data. Thus, ODE-RNN is introduced to solve this problem. Moreover, the introduced RNN structure shortens the integration range required for ODE Solver, helping reduce the error accumulation effect caused by the numerical solver.

\begin{figure}[t!]
	\centering
	\includegraphics[width=0.3\textwidth]{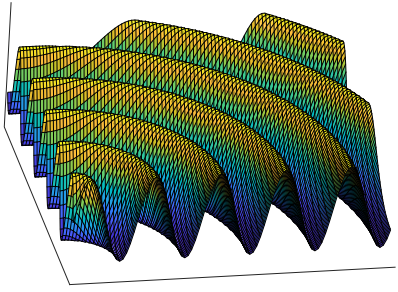}
		\vspace{-.5em}
	\caption{The value of path response (taking the real part as an example) in space, which shows highly nonsmooth spatially.}
		\vspace{-.5em}
	\label{path_response}
\end{figure}

\begin{figure}[htb!]
	\centering
		\vspace{-1em}
	\includegraphics[width=0.35\textwidth]{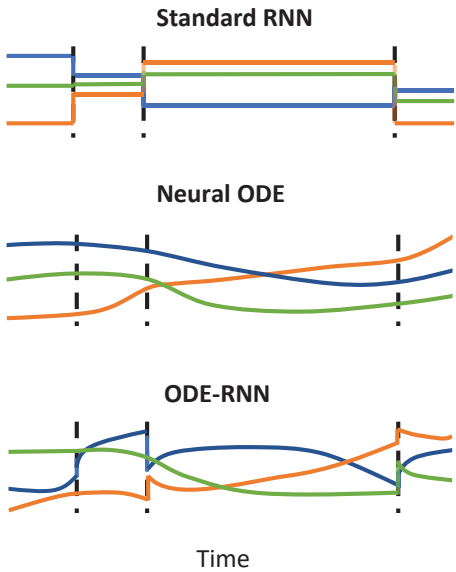}
		\vspace{-.5em}
	\caption{Hidden state trajectories. Standard RNNs have constant or undefined hidden states between observations. States of Neural ODE follow a complex trajectory but are determined by the initial state. The ODE-RNN model has states which obey an ODE between observations, and also updated at observations.}
		\vspace{-.5em}
	\label{ODE}
\end{figure}

\subsection{Why Conventional Network Dose NOT Work Well}
As shown in equation \eqref{Sye_CFR_eq1}, CSI matrix is determined by multipath response components interwinding. For a wireless channel, the variation mainly comes from two parts, i.e., the change of electromagnetic wave propagation paths caused by the changing of spatial position and frequency selective fading caused by Doppler effect. As shown in Fig. \ref{path_response}, due to the obvious magnitude difference of electromagnetic wavelength scale relative to the spatial changing scale of mobile users, the phase changing of channel response is very fast. Thus, the path response shows obvious non-smooth characteristics on the time axis and spatial coordinate axis. Combining the interwinding of multiple path responses and the high-dimensional characteristics of channel matrix caused by multi carriers and multi antennas, from the perspective of data characteristics, the channel prediction problem is essentially a prediction problem dealing with high-dimensional with extremely nonsmooth data.

As shown in Fig. \ref{ODE}, the hidden state of the standard RNN remains unchanged between observation points, so it is difficult to represent a complex changing process between observation points. Thus, the prediction accuracy of RNN greatly depends on the correlation between observation points and data smoothness. In other words, the network structure of RNN is difficult to deal with high-dimensional nonsmooth data such as CSI matrix. Reflected in the experiments, the prediction accuracy of RNN is particularly sensitive to the spatial or temporal interval of sequential sampling channels. When the time interval of acquisition channels is large or the spatial position between the two sampling channels is far, the prediction performance will be significantly degraded as shown in \cite{9569281} and our experiments in Section IV.

\subsection{Neural ODE for Representing the Physics Process of Channel Changing}

\begin{figure}[htb]
	\centering
	\includegraphics[width=0.5\textwidth]{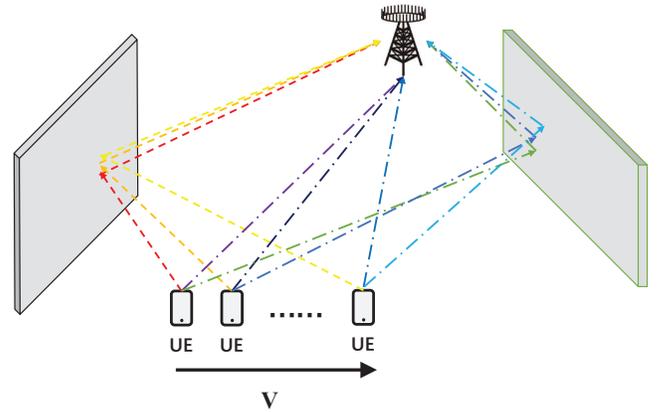}
		\vspace{-.5em}
	\caption{The changing process of propagation paths while the user is moving.}
		\vspace{-.5em}
	\label{path}
\end{figure}

Although CSI shows high-dimensional non-smoothness in data features, from the perspective of physics process of electromagnetic wave propagation, the changing of channel is a purely physics driving process. As shown in Fig. \ref{path}, in a real scenario, the BS serves a fixed area where there exist different static scatterers. Electromagnetic waves propagate from the transmitter to the receiver through line of sight, reflection, diffraction and other types of propagation processes. For each path, its path response can be decomposed into path attenuation and phase change. The path attenuation is related to parameters such as propagation distance and reflection coefficient, while the phase change is only related to the propagation distance. Considering that for a practical communication system, the time interval for acquiring CSI is the same as that of the coherent time slot, usually in the order of milliseconds, so the user can be regarded as performing uniform linear motion, i.e., the speed is constant during the considered coherent time slot.
Based on the above analysis, we provide the following theorem to show theoretical basis of using ODE-RNN for CSI prediction. 
\begin{figure*}
	\centering
	\includegraphics[width=1\textwidth]{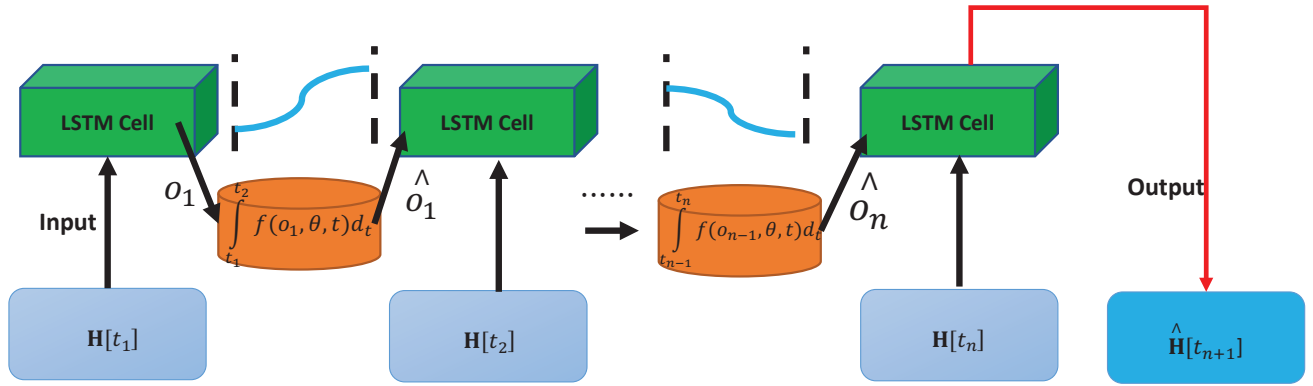}
		\vspace{-.5em}
	\caption{The learning structure of ODE-RNN.}
		\vspace{-.5em}
	\label{RNN}
\end{figure*}

\newtheorem{lemma}{\textbf{Theorem}}
\begin{lemma}
For a quasi-static scattering environment, the derivative of the mobile channel with respect to time is only related to the current CSI. 
\end{lemma}

\begin{proof}[\textbf{Proof 1:}]
As described above, the channel response of the user $u$ in position ${x_u}$ relative to the BS is composed of the multi-path responses. For propagation path $p$, the response is determined by the path delay, angle of arrival, reflection coefficient and other parameters. The above parameters are functions of geographical environment, electromagnetic characteristics of materials and position. Except for the user's location, other parameters are static parameters related to the scattering environment. Thus, there exists a mapping from user's position to CSI, i.e., 
\begin{equation}
g:\{ {x_u}\}  \to \{ {{\bf{H}}_u}\}.
\end{equation}
It should be noted that this mapping is actually bidirectional with high probability in practical wireless communication systems \cite{8292280}. Thus, there also exists a mapping from CSI to user's position, i.e., 
\begin{equation}
{g^{ - 1}}:\{ {{\bf{H}}_u}\}  \to \{ {x_u}\}.
\end{equation}
Therefore, we can pbtain 
\begin{equation}
{{{\bf{H}}(t + \Delta t) - {\bf{H}}(t)} \over {\Delta t}} = {{g\{ {g^{ - 1}}[{\bf{H}}(t)] + \Delta t \times {\bf{v}}\}  - {\bf{H}}(t)} \over {\Delta t}},
\end{equation}
where ${\bf{H}}(t)$ denotes the channel at the current time and ${\bf{v}}$ is the user's velocity vector which is independent of $t$. Therefore, this theorem is proved.
\end{proof}

Moreover, Theorem 1 can also be proved in the following way. 
\begin{proof}[\textbf{Proof 2:}]
According to equation (1),  the real and imaginary parts of channel response can be respectively wreitten as
\begin{equation}
\begin{aligned}
&{{\bf{h}}^{real}}[l] = \sum\limits_{p = 1}^K {{\bf{h}}_p^{real}}[l] \\&=  \sum\limits_{p = 1}^K {{\alpha _p}e({\theta _p})} \cos (2\pi [{{{d_p} + {v_u}\cos ({\theta _v} - {\theta _p}){\tau _p}} \over {{\lambda _l}}}]),
\end{aligned}
\end{equation}
and 
\begin{equation}
\begin{aligned}
	&{{\bf{h}}^{imag}}[l] = \sum\limits_{p = 1}^K {{\bf{h}}_p^{real}}[l] \\&=  \sum\limits_{p = 1}^K {{\alpha _p}e({\theta _p})} \sin (2\pi [{{{d_p} + {v_u}\cos ({\theta _v} - {\theta _p}){\tau _p}} \over {{\lambda _l}}}]),
\end{aligned}
\end{equation}
where ${{{\bf{h}}_p}}$ is the $p$th path response component. In particular, ${\alpha _p}$ is an inverse proportional function about ${d_p}$. Thus, we denote ${\alpha _p} = {{{\xi _p}} \over {{d_p}}}$, where ${\xi _p}$ is only related to the characteristics of electromagnetic materials. Also, we have ${\tau _p} = {{{d_p}} \over c}$, where $c$ is the speed of light. Thus, we have $2\pi [{{{d_p} + {v_u}\cos ({\theta _v} - {\theta _p}){\tau _p}} \over {{\lambda _l}}}] = {\rho_p} {d_p}$. Taking the real part of channel response as an example, its derivative with respect to time is
\begin{equation}
\begin{aligned}
&{{\partial {{\bf{h}}^{real}}[l]} \over {\partial t}} = {{\partial {{\bf{h}}^{real}}[l]} \over {\partial {d_p}}}{{\partial {d_p}} \over {\partial t}}\\& = [\sum\limits_{p = 1}^K {{{{\xi _p}} \over {{d_p}}}{\bf{e}}({\theta _p})} \sin ({\rho _p}{d_p}){\rho _p} - \sum\limits_{p = 1}^K {{{{\xi _p}} \over {{d_p}}}{\bf{e}}({\theta _p})} \cos ({\rho _p}{d_p}){1 \over {{d_p}}}]{{\partial {d_p}} \over {\partial t}} \\& = \sum\limits_{p = 1}^K {({\rho _p}} {\bf{h}}_p^{imag}[l] - {1 \over {{d_p}}}{\bf{h}}_p^{real}[l]){{\partial {d_p}} \over {\partial t}}.
\end{aligned}
\end{equation}
Similarly, the derivative of the imaginary part can be given by 
\begin{equation}
{{\partial {{\bf{h}}^{imag}}[l]} \over {\partial t}} =  - \sum\limits_{p = 1}^K {({\rho _p}{\bf{h}}_p^{real}[l] + {1 \over {{d_p}}}{\bf{h}}_p^{imag}[l])} {{\partial {d_p}} \over {\partial t}}.
\end{equation}
Considering that the velocity vector hardly changes during such short coherent time slot, 
${{\partial {d_p}} \over {\partial t}}$ is a function decided by the user's position at time $t$, which is also decided by the current CSI according to what we analysed in Proof 1. Thus, the whole formula only has  one variable ${\bf{h}}[l]$. Therefore, this theorem is proved.
\end{proof}

Consequently, combined with the previous analysis, the physics process of channel changing is quite suitable to be characterized by Neural ODE. Thus, a high-dimensional nonsmooth prediction problem is transformed into a continuous physics driven forward calculation process.

\subsection{Learning Structure of ODE-RNN}
In addition to the inherent advantages of ODE-RNN over Neural ODE network mentioned in the previous subsection, there exists another major consideration in applying ODE-RNN to mobile channel prediction task. In the practical scenario, it is inevitable that some component paths maybe untable and scattering environment can be slightly disturbed. In this case, there exists some errors in the values of some observation points which may cause obvious error accumulation when using Neual ODE structure. At each observation point, ODE-RNN will use the current observation value and the hidden state value propagated forward to calculate the current hidden state value. This learning structure greatly enhances the adaptability and robustness of the system in practical scenarios.

The learning structure of ODE-RNN is shown in Fig. \ref{RNN}. The real part and imaginary part of the complex valued channel matrix are spliced into a real matrix. RNN adopts LSTM mechanism to increase its learning ability. It is worth mentioning that this network structure can adapt to both equal interval sampling and unequal interval sampling cases. When unequal interval sampling scheme is adopted, the sampling time interval needs to be used as the input of Neural ODE to adjust the integration length of ODE solver in forward calculation.

\begin{figure}[b]
	\centering
	\includegraphics[width=0.4\textwidth]{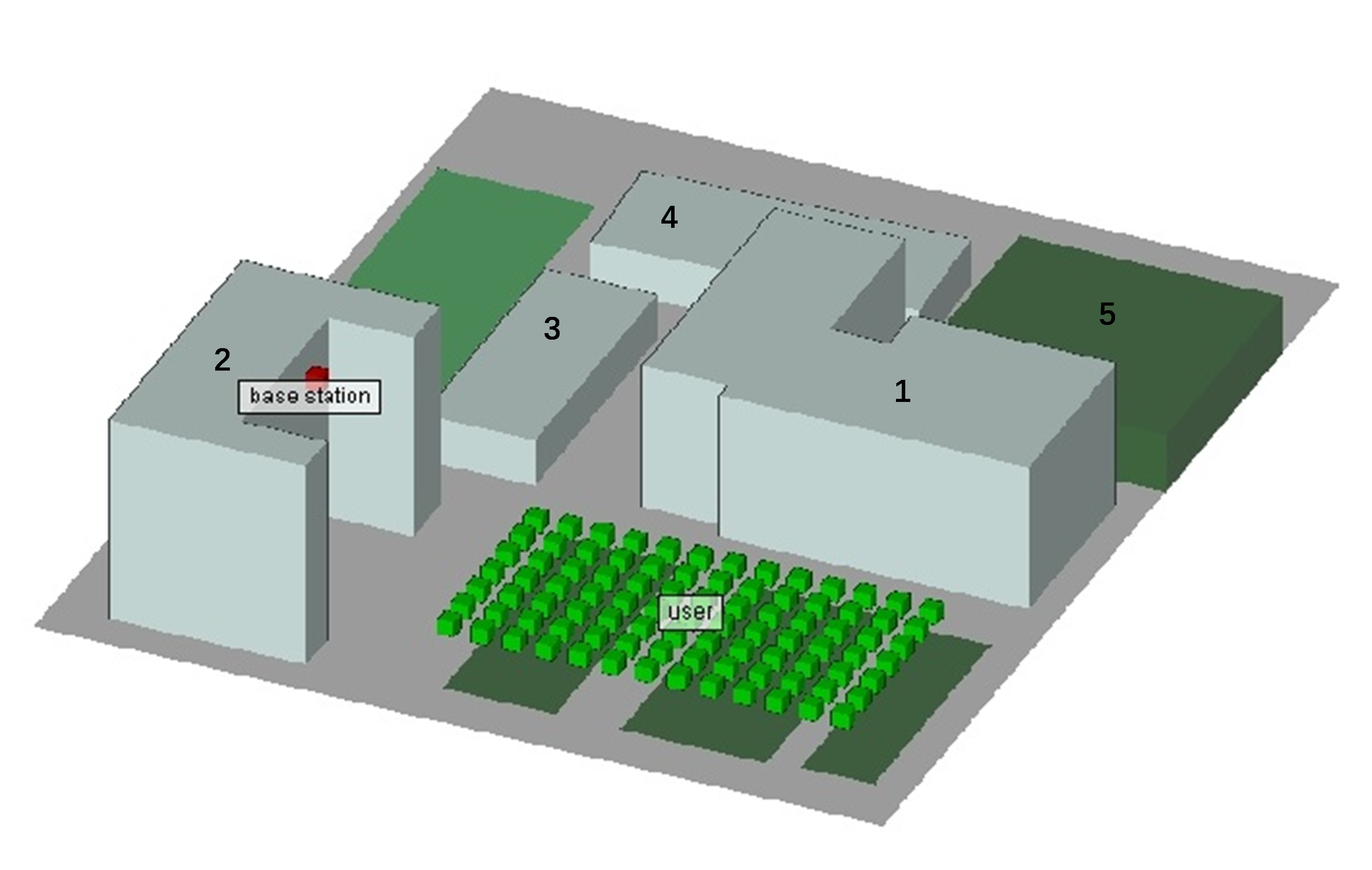}
		\vspace{-.5em}
	\caption{The 3D model of the ray-tracing scene.}
	\label{model}
		\vspace{-.5em}
\end{figure}

\section{Performance Evaluation} 
\subsection{Scene Setup and Datasets Generation} \label{scene}
As shown in Fig. \ref{model}, we choose a practical outdoor scenario to set up our 3D model. Wireless Insight from Remcom company is used to do the ray-tracing calculation. In Fig. \ref{model}, the height of Building 1 is 25m, the height of Building 2 is 35m, the height of building 3 and building 4 is 8m and the building material is cement. Area 5 is a forest. Users are distributed in a 120m × 60m area. The central frequency is set to 3.5GHz. The BS is located 10m higher above building 2 and is equipped with a ULA. The OFDM bandwidth is 100MHz and the maximum number of paths is 25.
Considering that the channel sampling time interval is small in practical system, it can be assumed that the user moves in a uniform linear way in any direction within the length of sequence time. The corresponding Doppler phase shift is calculated and applied to each propagation path. The channel CSI matrix is calculated by ray-tracing algorithm.

\begin{table}[b!]
	\caption{Main Parameters and Values for LSTM network.}
	\begin{center}
		\begin{tabular}{ p{4cm}   p{4cm}}
			\toprule
			\textbf{Parameters} & \textbf{Value} \\
			\toprule
			Input dimension & 64×64×2\\
			
			Output dimension & 64×64×2 \\
			
			Activation function & Tanh \\
			
			Number of neurons in hidden layer & 384 \\			
			
			Performance metric & Mean square error (MSE) \\
			
			Optimizer & Adam \\
			
			Training steps & $2 \times {10^5}$ \\
			
			Learning rate & $1 \times {10^{ - 3}}$ \\
			
			Batch size & 20 \\
			
			Training samples & $80\% $ of the whole datasets \\
			\toprule
		\end{tabular}		
	\end{center}
	\label{LSTM_setting}
\end{table}

\begin{table}[b!]
	\caption{Main Parameters and Values for Neural ODE}
	\begin{center}
		\begin{tabular}{ p{4cm}   p{4cm}}
			\toprule
			\textbf{Parameters} & \textbf{Value} \\
			\toprule
			Network type & MLP\\
			
			Iutput dimension & 384 \\
			
			Activation function & Tanh   \\
			
			Number of neurons in hidden layer & 512-1024-512 \\			
			
			ODE Solver & Adaptive Solver \\
			
			Backpropagation & Adaptive Checkpoint\\
			& Adjoint Method \\
			\toprule
		\end{tabular}		
	\end{center}
	\label{MLP_setting}
\end{table}

\subsection{Experimental Results} \label{performance}
To evaluate our proposed approach comprehensively, we compare our proposed method directly with existing works. The comparison includes the prediction accuracy under different user’s speeds and different sequence lengths. Moreover, in order to show the adaptability of the proposed method to the real environment, the comparison also include the robustness of the networks. Firstly, to verify the effect of using Neural ODE network, the experiment will compare our proposed method with sequential learning structure. In order to ensure the fairness of the comparison, the number of parameters of the two networks will be approximately consistent with our network. Secondly, in order to verify the necessity of integrating Neural ODE with sequence learning structure, the prediction performance of normal Neural ODE and ODE-RNN network will be compared. In addition, the prediction performance of the network with channel noise is considered to verify the robustness of the proposed network. For the reason that the proposed method does not require equal sampling interval, when the channel quality is poor, ODE-RNN can discard part of the observation results while still work. We will also evaluate the performance of this scheme.

The parameter settings for ODE-RNN are shown in Tables \ref{LSTM_setting} and \ref{MLP_setting}. TensorFlow library \cite{Mart2016TensorFlow} is used to train the networks. Mean square error (MSE) is used as the loss function and MSE can be mathematically written as
\begin{equation}
\text{MSE} = {1 \over i}{\sum\limits_{m = 1}^i {({y_m} - \mathop {{y_m}}\limits^ \wedge  )} ^2},
\end{equation}
where $i$ is the dimension of output, $y$ is the training label and ${\mathop y\limits^ \wedge}$ is the prediction value. It represents the average distance between the target value and prediction value of all the output dimensions.

Normalized MSE (NMSE) is used to evaluate the prediction accuracy of the testing datasets, which can be written as
\begin{equation}
\text{NMSE} = \mathbb{E}\left({{\sum\limits_{m = 1}^i {|{y_m} - {{\mathop {{y_m}|}\limits^ \wedge  }^2}} } \over {\sum\limits_{m = 1}^i {|{y_m}{|^2}} }}\right).
\end{equation}
NMSE is an expectation value calculated across the testing dataset. Compared with MSE, NMSE is more convenient for comparison crossing different datasets. 

\begin{figure}
	\centering
	\includegraphics[width=0.45\textwidth]{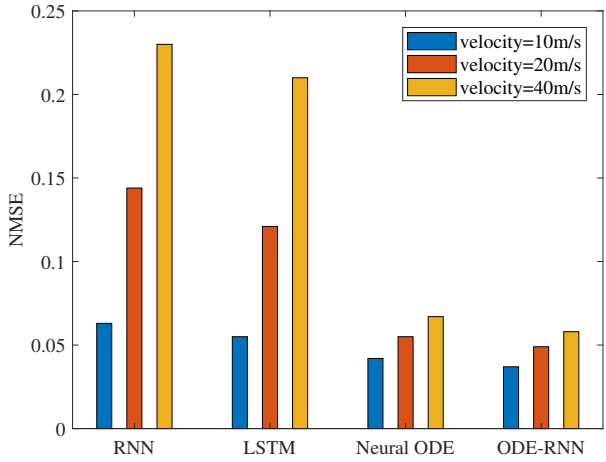}
		\vspace{-.5em}
	\caption{The prediction NMSE comparison among different methods with various user velocity.}
		\vspace{-.5em}
	\label{ex1}
\end{figure}
\begin{figure}
	\centering
	\includegraphics[width=0.45\textwidth]{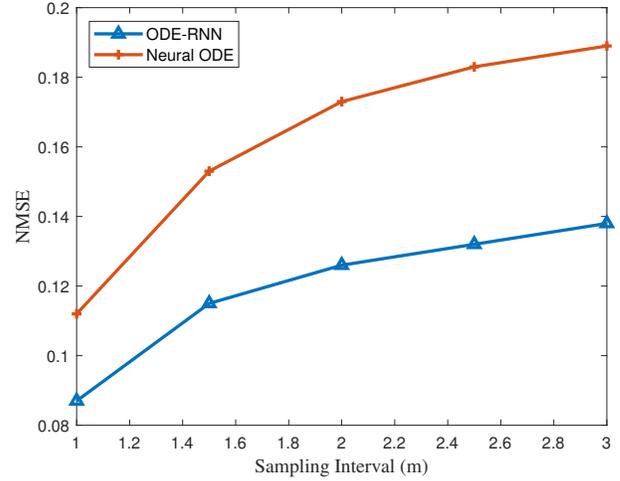}
		\vspace{-.5em}
	\caption{The prediction NMSE comparison between ODE-RNN and Neural ODE when the sampling intervals are much higher.}
		\vspace{-.5em}
	\label{ex2}
\end{figure}
\begin{figure}
	\centering
	\includegraphics[width=0.45\textwidth]{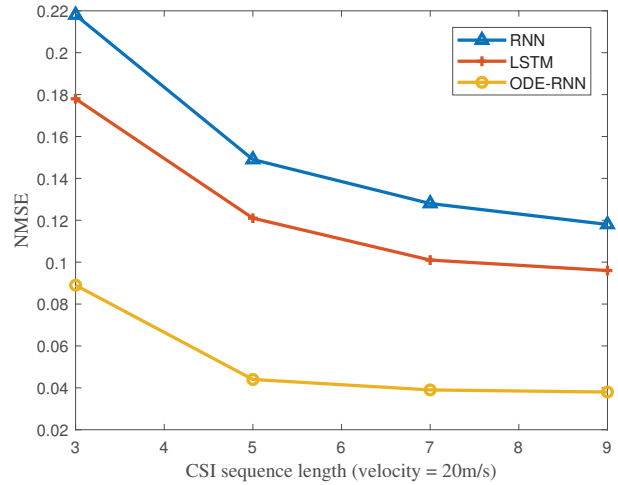}
		\vspace{-.5em}
	\caption{The prediction NMSE comparison between different methods under different CSI sequence length.}
	\label{ex3}
		\vspace{-.5em}
\end{figure}
\begin{figure}
	\centering
	\includegraphics[width=0.45\textwidth]{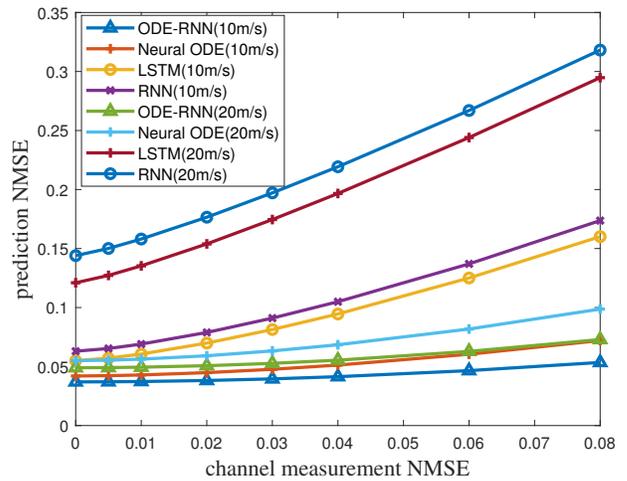}
	\vspace{-.5em}
	\caption{The prediction NMSE under different channel measurement error with different method and user’s velocity.}
		\vspace{-.5em}
	\label{ex4}
\end{figure}
The prediction NMSE results compare between different methods and different user's velocity is shown in Fig. \ref{ex1}. For fair comparison, all sequence lengths are set to 5. It can be seen from the results that the overall performance of the methods adopting Neural ODE network is better than the existing methods. More importantly, the performance of existing works based on RNN network structure is very sensitive to the change of user's speed. The main reason is that the performance of RNN is extremely dependent on the correlation between adjacent observations. When the user moves faster, the data correlation between adjacent sampling points obviously decreases. Therefore, the performance of RNN and LSTM decreases significantly. However, the network based on ODE structure has obvious performance advantages for faster speed because it hardly relies on the data correlation between observations.

In order to better illustrate the advantages of using ODE-RNN structure over ordinary Neural ODE, Fig. \ref{ex2} compares the prediction performance of the two approaches with greater spatial spacing between adjacent sampling points. It is worth noting that under such sampling interval, RNN and LSTM have been difficult to converge and apply due to the poor correlation between adjacent data. So, the performance of RNN and LSTM will not be added to the comparison. For the reason that ODE-RNN can take use of different observations on the final prediction results in the forward calculation and the integration distances between different time points are short, the result shows that ODE-RNN enjoys significant performance advantage over Neural ODE.

Fig. \ref{ex3} shows the prediction accuracy of the different methods under various sequence lengths. It is worth mentioning that for the normal Neural ODE, the forward calculation only depends on its initial input value after the network is trained and parameters are fixed. As a result, Neural ODE is not compared in this figure. According to Fig. \ref{ex3}, it is shown that compared with the existing methods, ODE-RNN can obtain the best performance with less sequence length which enables the system to have a wider range of sequence lengths to choose for adapting to different application cases, making the system more flexible in practice.

In order to show the advantages of the proposed approach in network robustness and adaptiveness, the influence of channel sequence with measurement error in the inferring stage on prediction performance is considered. The experimental result is shown in Fig. \ref{ex4}. Here, assuming that the channel noise obey Gaussian distribution of zero means. The average NMSE of the noisy channel w.r.t the real channel represents the channel quality. Because that Neural ODE does not require equal interval of channel sequences, some high deviation observations can be discarded, while the existing methods must retain those data. In fact, for RNN networks, a small amount of high deviation observation data will greatly destroy the sequence correlation, and then affect the accuracy of prediction. It can be seen that the method proposed by us can simply avoid this kind of problem, and is far more robust and adaptive than the existing methods.

\section{Conclusions} \label{conclusion}
In this paper, a physics inspired adaptive channel prediction method based on ODE-RNN is proposed. Different from the existing work focusing on correlation between sequence data, the network structure proposed in this paper is designed motivated by the continuous physics changing process of channel response. The proposed network can predict the changing of CSI in mobile environment with high accuracy. Numerical results verified that the proposed method can bring a significant improvement in prediction accuracy, system flexibility and network robustness compared with existing methods.

\bibliographystyle{IEEEtran}
\bibliography{bibfile}

\begin{thebibliography}{10}
\providecommand{\url}[1]{#1}
\csname url@samestyle\endcsname
\providecommand{\newblock}{\relax}
\providecommand{\bibinfo}[2]{#2}
\providecommand{\BIBentrySTDinterwordspacing}{\spaceskip=0pt\relax}
\providecommand{\BIBentryALTinterwordstretchfactor}{4}
\providecommand{\BIBentryALTinterwordspacing}{\spaceskip=\fontdimen2\font plus
\BIBentryALTinterwordstretchfactor\fontdimen3\font minus
  \fontdimen4\font\relax}
\providecommand{\BIBforeignlanguage}[2]{{%
\expandafter\ifx\csname l@#1\endcsname\relax
\typeout{** WARNING: IEEEtran.bst: No hyphenation pattern has been}%
\typeout{** loaded for the language `#1'. Using the pattern for}%
\typeout{** the default language instead.}%
\else
\language=\csname l@#1\endcsname
\fi
#2}}
\providecommand{\BIBdecl}{\relax}
\BIBdecl

\bibitem{9427230}
C.~Wu, X.~Yi, Y.~Zhu, W.~Wang, L.~You, and X.~Gao, ``Channel prediction in
  high-mobility massive {MIMO}: From spatio-temporal autoregression to deep
  learning,'' \emph{IEEE Journal on Selected Areas in Communications}, vol.~39,
  no.~7, pp. 1915--1930, 2021.

\bibitem{8395053}
C.~Luo, J.~Ji, Q.~Wang, X.~Chen, and P.~Li, ``Channel state information
  prediction for {5G} wireless communications: A deep learning approach,''
  \emph{IEEE Transactions on Network Science and Engineering}, vol.~7, no.~1,
  pp. 227--236, 2020.

\bibitem{6945858}
R.~O. Adeogun, P.~D. Teal, and P.~A. Dmochowski, ``Extrapolation of {MIMO}
  mobile-to-mobile wireless channels using parametric-model-based prediction,''
  \emph{IEEE Transactions on Vehicular Technology}, vol.~64, no.~10, pp.
  4487--4498, 2015.

\bibitem{2000Long}
A.~Duel-Hallen and S.~Hu, ``Long-range prediction of fading signals,''
  \emph{IEEE Signal Process Mag}, vol.~17, no.~3, pp. 62--75, 2000.

\bibitem{2002Recurrent}
J.~T. Connor, R.~D. Martin, and L.~E. Atlas, ``Recurrent neural networks and
  robust time series prediction,'' \emph{IEEE Transactions on Neural Networks},
  vol.~5, no.~2, pp. 240--254, 2002.

\bibitem{2014Fading}
T.~Ding and A.~Hirose, ``Fading channel prediction based on combination of
  complex-valued neural networks and chirp z-transform,'' \emph{IEEE
  Transactions on Neural Networks and Learning Systems}, vol.~25, no.~9, pp.
  1686--1695, 2014.

\bibitem{2020Recurrent}
J.~Wei and H.~D. Schotten, ``Recurrent neural networks with long short-term
  memory for fading channel prediction,'' in \emph{2020 IEEE 91st Vehicular
  Technology Conference (VTC2020-Spring)}, 2020.

\bibitem{8904286}
Y.~Huangfu, J.~Wang, R.~Li, C.~Xu, X.~Wang, H.~Zhang, and J.~Wang, ``Predicting
  the mumble of wireless channel with sequence-to-sequence models,'' in
  \emph{2019 IEEE 30th Annual International Symposium on Personal, Indoor and
  Mobile Radio Communications (PIMRC)}, 2019, pp. 1--7.

\bibitem{2018Optimization}
P.~Stapor, F.~Frhlich, and J.~Hasenauer, ``Optimization and uncertainty
  analysis of ode models using 2nd order adjoint sensitivity analysis,'' 2018.

\bibitem{2020Adaptive}
J.~Zhuang, N.~Dvornek, X.~Li, S.~Tatikonda, X.~Papademetris, and J.~Duncan,
  ``Adaptive checkpoint adjoint method for gradient estimation in neural ode,''
  \emph{arXiv e-prints}, 2020.

\bibitem{9569281}
Z.~Xiao, Z.~Zhang, C.~Huang, C.~Zhong, and X.~Chen, ``Gpae-lstmnet: A novel
  learning structure for mobile {MIMO} channel prediction,'' in \emph{2021 IEEE
  32nd Annual International Symposium on Personal, Indoor and Mobile Radio
  Communications (PIMRC)}, 2021, pp. 1--6.

\bibitem{8292280}
J.~Vieira, E.~Leitinger, M.~Sarajlic, X.~Li, and F.~Tufvesson, ``Deep
  convolutional neural networks for massive {MIMO} fingerprint-based
  positioning,'' in \emph{2017 IEEE 28th Annual International Symposium on
  Personal, Indoor, and Mobile Radio Communications (PIMRC)}, 2017, pp. 1--6.

\bibitem{Mart2016TensorFlow}
M.~Abadi, P.~Barham, J.~Chen, Z.~Chen, and X.~Zhang, ``Tensorflow: A system for
  large-scale machine learning,'' 2016.

\end{thebibliography}
\end{document}